\documentclass[10pt,twocolumn,twoside]{IEEEtran}
\IEEEoverridecommandlockouts
\usepackage{setspace}
\usepackage{graphicx}
\usepackage{url}
\usepackage{xcolor}
\usepackage{amsmath}
\usepackage{mathtools}
\usepackage{tikz}
\usepackage{verbatim}
\usepackage{subcaption}
\usepackage{pgfplots}
\pgfplotsset{compat=1.10}
\usepackage{graphicx}
\usepackage{epstopdf}
\usepackage{amssymb}
\usepackage{relsize}
\usepackage[textwidth=3.8em,textsize=scriptsize,disable]{todonotes}
\usepackage{cite}
\usepackage{algorithm}
\definecolor{Gray}{gray}{0.90}
\usepackage{colortbl}

 \usepackage{url}
 \usepackage{hyperref}
\usepackage{algorithm}
\usepackage{algorithmicx}
\usepackage{algcompatible}
\usepackage{algpseudocode}
\usepackage{color}
\usepackage{blkarray}
\usepackage{multirow}
\usepackage{algpseudocode}
\usepackage{bm}
\usepackage{amsthm}
\newtheorem{theorem}{\bf{Theorem}}[section]

\newtheorem{cor}[theorem]{Corollary}
\newtheorem{lem}[theorem]{\bf{Lemma}}

\newtheorem{definition}{Definition}[section]

\usepackage{amssymb}
\usepackage{url}
\usepackage{enumitem}
\newlist{Properties}{enumerate}{2}
\setlist[Properties]{label=Property \arabic*., font=\textbf, itemindent=*}

\begin{document}
\title{\LARGE{A Geometric Approach to Resilient Distributed Consensus Accounting for State Imprecision and Adversarial Agents}}
\author{Christopher A. Lee and Waseem Abbas
\thanks{C.~Lee is with the Electrical Engineering Department, and W.~Abbas is with the Systems Engineering Department at the University of Texas at Dallas, Richardson, TX, USA (Emails: \texttt{christopher.lee3@utdallas.edu, waseem.abbas@utdallas.edu}).}
}

\maketitle
\begin{abstract}
This paper presents a novel approach for resilient distributed consensus in multiagent networks when dealing with adversarial agents imprecision in states observed by normal agents. Traditional resilient distributed consensus algorithms often presume that agents have exact knowledge of their neighbors' states, which is unrealistic in practical scenarios. We show that such existing methods are inadequate when agents only have access to imprecise states of their neighbors. To overcome this challenge, we adapt a geometric approach and model an agent's state by an `imprecision region' rather than a point in $\mathbb{R}^d$. From a given set of imprecision regions, we first present an efficient way to compute a region that is guaranteed to lie in the convex hull of true, albeit unknown, states of agents. We call this region the \emph{invariant hull} of imprecision regions and provide its geometric characterization. Next, we use these invariant hulls to identify a \emph{safe point} for each normal agent. The safe point of an agent lies within the convex hull of its \emph{normal} neighbors' states and hence is used by the agent to update it's state. This leads to the aggregation of normal agents' states to safe points inside the convex hull of their initial states, or an approximation of consensus. We also illustrate our results through simulations. Our contributions enhance the robustness of resilient distributed consensus algorithms by accommodating state imprecision without compromising resilience against adversarial agents.
\end{abstract}


\section{Introduction}
\label{sec:intro}
In distributed multi-agent networks, agents work collaboratively to accomplish complex tasks by sharing information with each other. To make optimal decisions, they rely on this shared data. However, false or misleading information from a subset of agents--due to adversarial actions or other failures--can significantly compromise the network's operations. Thus, to realize the full benefits of distributed decision-making, resilience to abnormal agents is crucial for the correct functioning of multi-agent systems. As such, there has been an increased interest in characterizing and imposing resilience in distributed systems.

Recent works on resilient distributed networks, particularly in resilient distributed consensus—a canonical problem in networked and distributed control systems--present a suite of algorithms and conditions on the network structure. If these conditions are satisfied, they guarantee the normal operations of the overall system~\cite{ishii2022overview,pirani2023graph}. The main idea of such algorithms is to minimize the impact of information from adversarial or faulty agents, even without knowing their identities~\cite{park2017fault,abbas2022resilient,leblanc2013resilient,su2016multi,mendes2015multidimensional,yan2022resilient,yan2020safe,wang2022resilient,li2019resilient,yang2019byrdie,su2020byzantine,zhu2023resilient,kuwaranancharoen2020byzantine,usevitch2019resilient,safi2022resilient,dibaji2017resilient}. However, existing resilient consensus solutions often make the unrealistic assumption that agents can observe their neighbors' exact, unperturbed states. In real-world applications, agents deal with imprecise data due to factors like sensor noise, environmental conditions, or hardware limitations~\cite{jayasimha1994fault,nakamura2007information,loffler2009data,park2017security,frehse2014formal}. As we demonstrate in this paper, applying current resilient algorithms in these `imprecise' scenarios can lead to failures, even if the number of adversarial agents is below the theoretical threshold allowed by the solution. Therefore, new methods are needed, and this paper introduces strategies to handle both adversarial agents and imprecise states in multi-agent networks.

This paper introduces a new method to tackle the resilient distributed consensus problem, focusing on the commonly overlooked issue of state imprecision, that is, the observed states vary from the true state by a known bounded measure. There are existing studies that address a similar problem in a static scenario \cite{Takayuki2000}, however our formulation is developed independently and tailored towards the application of resilient consensus. Our approach ensures that all normal agents remain within the convex hull of their initial states and continuously converge towards a ball contained in the convex hull of the normal agents' initial positions.  The radius of the ball is dependent on the initial configuration, the actions of the attacker, and the stochastic discrepancy between true and observed states.  However, the final state of each normal agent is guaranteed to be strictly contained in the initial convex hull of normal agents. We show that this approximation of ``consensus" is achieved by the CPIH algorithm proposed in this paper, while existing distributed consensus algorithms fail to achieve even approximate consensus under imprecision. The crux of our method is ensuring that in each step of the state update process, a normal agent can find a point that is guaranteed to lie in the convex hull of its normal neighbors only. The agent then updates its state by moving towards that point, which is also referred to as the \emph{safe point}~\cite{abbas2022resilient,park2017fault,yan2020safe,li2022byzantine}. We show that existing methods fail to identify such safe points when faced with imprecision. In contrast, our approach successfully identifies these points and maintains the same level of resilience against adversarial agents as existing methods that do not account for state imprecision. 

We summarize our key contributions as follows:

\begin{itemize}
    \item First, we show that current algorithms for resilient distributed consensus do not work well when agents have imprecise states (Section~\ref{sec:imprecision}). This is important because imprecise states are common in real-world applications.

\item Second, we introduce a new way to model these imprecise states. Instead of using points in $\mathbb{R}^d$, we use what we call `imprecision regions' containing true states of agents. We then introduce the notion of \emph{invariant hull} which is essentially the largest region that is contained in the convex hull of normal agents' true states when only their imprecision regions are known. We also provide a geometric characterization of invariant hulls and an efficient algorithm to compute them (Section~\ref{sec:Invariant_hull}).

    \item Third, we develop a method that leverages these invariant hulls to identify a `safe point' in the neighborhood of a normal agent. This neighborhood may contain up to $\frac{N_v}{d+1} -1 $ adversaries. Here, $N_v$ is the total number of nodes in the neighborhood of a normal agent $v$ and $d$ is the state dimension. This safe point allows the normal agent to update its state while ignoring any adversaries. We also illustrate our results through simulations (Section~\ref{sec:CPIH}).
\end{itemize}

The rest of the paper is organized as follows: Section~\ref{sec:prelim} introduces preliminaries and notations. Section~\ref{sec:imprecision} reviews existing resilient consensus algorithms and show their inadequacy with imprecise states. Section~\ref{sec:Invariant_hull} introduces the invariant hull notion and discusses its computation and significance. Section~\ref{sec:CPIH} presents methods to compute safe point for resilient distributed consensus with imprecision, and present simulations. Finally, Section~\ref{sec:con} concludes the paper. 
\section{Preliminaries}
\label{sec:prelim}
We consider a network of agents modeled by an undirected graph $G = (V,E)$, where $V$ represents agents and $E$ denotes interactions among agents. An edge between agents $u$ and $v$ is denoted by an unordered pair $(u,v)$. We use the terms \emph{agent} and \emph{node} interchangeably. Each agent $u\in V$ has a $d$-dimensional state vector whose value is updated over time. The state of agent $u$ {at time $t$} is represented by {a point $x_u(t) \in \mathbb{R}^d$}. The \emph{neighborhood} of $u$ is the set of nodes $N_u =~\{v\in V: \; (u,v)\in E\}\cup\{u\}$ (node $u$ is included in its neighborhood). For a given set of points ${X}\subset\mathbb{R}^d$, we denote its \emph{convex hull} by \texttt{Conv}$({X})$. Also, we use terms \emph{points} and \emph{states} interchangeably.

\emph{Normal and Adversarial Agents --} Agents in the network can either be {normal} or {adversarial}. \emph{Normal} agents, denoted by $V_n\subseteq V$, are the ones that interact with their neighbors synchronously and update their states according to a pre-defined state update rule, which is the consensus algorithm. \emph{Adversarial} agents, denoted by $V_f\subset V$, are the ones that can change their states arbitrarily. Moreover, an adversarial node can transmit different values to its different neighbors, which is referred to as the \emph{Byzantine} model. The number of adversarial nodes in the neighborhood of a normal node $u$ is denoted by $F_u$. For a normal node $u$, all nodes in its neighborhood are indistinguishable, that is, $u$ cannot identify which of its neighbors are adversarial.


\emph{Agent Imprecision --} We consider that each agent has an imprecise information of its neighbor's state. In other words, an agent $u$ having agent $v$ as a neighbor, acquires/observes an imprecise value of agent $v$'s state, which could be due to perturbation as a result of improper calibration, hardware imprecision, or other measurement uncertainties. We denote this imprecise value of agents $v$'s state by $r_v$, and associate a \emph{region of imprecision}, $B_v\subset \mathbb{R}^d$ such that the true value of agent $v$'s state, i.e., $x_v$, can be anywhere inside this $B_v$. An example of this imprecision model is a hypersphere with a center $r_v$ and radius $\delta$. The true state value $x_v$ can be anywhere inside this hypersphere of radius $\delta$ (indicating the extent of imprecision). Another example of imprecision model is a hypercube, where imprecision in each coordinate is at most $\delta$, as shown in Figure~\ref{fig:imprecision_box}. 

\begin{figure}[!h]
\centering
\includegraphics[scale=0.75]{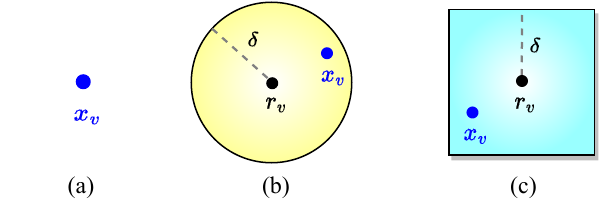}
\caption{(a) An exact point $x_v$ (no imprecision). (b) A disk imprecision region in $\mathbb{R}^2$. The observed point is $r_v$, which is the imprecise version of the exact point $x_v$ that lies somewhere inside the disk. (c) A square imprecision region in $\mathbb{R}^2$.}
\label{fig:imprecision_box}
\end{figure}


\emph{Resilient Consensus Problem --} In a network with both normal and adversarial agents, the aim of resilient vector consensus is to design a mechanism that ensures all normal agents update their states to eventually converge on a common state. This common state must lie within the convex hull of their initial states, denoted as $X(0) = \{x_1(0),x_2(0),\cdots,x_n(0)\}$. The mechanism must satisfy the following conditions: 

\begin{itemize}
    \item \emph{Safety:} At any time step $t$, the state of any normal node $v$ must lie within $\texttt{Conv}(X(0))$.
    \item \emph{Agreement:} For every $\epsilon >0$, there exists a time $t_\epsilon$, such that the states of any two normal agents $u$ and $v$ come within $\epsilon$ of each other for all $t> t_\epsilon$. 
\end{itemize}

Next, we discuss the resilient consensus solutions with and without imprecision. 
\section{Effect of Imprecision}
\label{sec:imprecision}
In this section, first, we review a resilient distributed consensus algorithm guarantees convergence of normal nodes despite adversarial agents without considering any imprecision in agents' states. Then, we consider imprecision and show that the existing resilient algorithms fail to achieve consensus and perform poorly under states' imprecision. 
\subsection{Resilient Consensus with No Imprecision}
Several resilient consensus schemes have been proposed in the literature for the scalar and higher dimensional cases, e.g., \cite{park2017fault,abbas2022resilient,leblanc2013resilient,su2016multi,mendes2015multidimensional,yan2022resilient,yan2020safe,wang2022resilient,ishii2022overview,li2019resilient,yang2019byrdie,su2020byzantine,zhu2023resilient,kuwaranancharoen2020byzantine,pirani2023graph}. The main approach, especially in higher dimensions $d\ge 2$, relies on the principle that in each round/step of the consensus algorithm, each normal node computes a point that lies in the convex hull of its normal neighbors' state, and then updates its state by moving towards that point, which is sometimes referred to as the \emph{`safe point'}. The computation of this safe point is a challenging endeavour, and various methods are employed. A useful way that results in superior resilience (in terms of the number of adversarial agents that can be allowed) is to obtain a safe point by computing a centerpoint of the agents' states. A centerpoint is a generalization of median in higher dimensions, and is defined below. 

\begin{definition}(Centerpoint)
Given a set \(S\) of \(N\) points in \(\mathbb{R}^d\), a centerpoint \(p\) is a point with the property that every closed halfspace of $\mathbb{R}^d$ containing \(p\) must also contain at least \( \frac{N}{d+1} \) points of \(S\). 
\end{definition}

The set of all centerpoints is referred to as the \emph{centerpoint region}. Figure~\ref{fig:CP} illustrates an example. There are six points in $\mathbb{R}^2$, and any line passing through a centerpoint divides these six points into two regions, each containing at least two points, as in Figures~\ref{fig:CP}(a) and (b). Figure~\ref{fig:CP}(c) illustrates the centerpoint region for the given example.
\begin{figure}[!h]
\centering
\begin{subfigure}[b]{0.15\textwidth}
\centering
\includegraphics[scale=0.82]{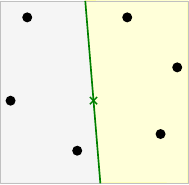}\\
\caption{}
\end{subfigure}
\begin{subfigure}[b]{0.15\textwidth}
\centering
\includegraphics[scale=0.82]{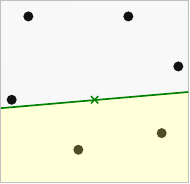}
\caption{}
\end{subfigure}
\begin{subfigure}[b]{0.15\textwidth}
\centering
\includegraphics[scale=0.82]{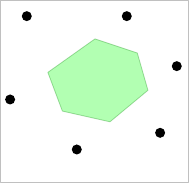}
\caption{}
\end{subfigure}
\caption{Illustration of centerpoint. In (a) and (b), centerpoint is denoted by `$\times$' and lines are passing through the centerpoint. The green shaded region in (c) is the centerpoint region. }
\label{fig:CP}
\end{figure}
The notion of centerpoint provides a complete characterization of the safe point. It is shown in \cite{abbas2022resilient} that a safe point for an agent $v$ is essentially a centerpoint of its neighbors' states as long as the number of adversaries in the neighborhood of $v$ is bounded by $\frac{N_v}{d+1}$. Here, $d$ is the dimension of the state and $N_i$ is the number of neighbors of $v$. Moreover, a safe point may not exist if the number of adversaries is more than $\frac{N_v}{d+1}$. Thus, the centerpoint-based safe point computation is particularly useful. 

Now, a \emph{resilient consensus algorithm} based on centerpoints and considering \emph{no imprecision}---agents observe their normal neighbors' states exactly---can be designed as follows:

\begin{itemize}
    \item In each iteration $t$, a normal agent $u$ gathers the state values of its neighbors in $N_u$, and computes a safe point $s_u(t)$ by computing a centerpoint agents' states.
    \item Agent $u$ updates its states as follows:
    \begin{equation}
        \label{eq:res_con}
        x_u(t+1) = \alpha_u(t)s_u(t) + (1-\alpha_u(t))x_u(t),
    \end{equation}
    where $\alpha_u(t)\in (0\;\;1)$ is a dynamically chosen parameter whose value depends on the application \cite{park2017fault}.
\end{itemize}

Analysis in \cite{abbas2022resilient} and \cite{park2017fault} shows that the above scheme achieves resilient consensus as long as the number of adversaries in each normal agent $u$ is at most $\frac{N_u}{d+1}-1$. 

Next, we see how the above resilient consensus (and other safe point based) algorithms perform if we consider imprecise agents' states.

\subsection{Resilient Solutions Do Not Work Under Imprecision}
\label{sec:Res_do_not_work}
The fundamental premise of resilient consensus algorithms lies in computing safe points that are always inside the convex hull of \emph{normal agents' true states}. In the case of imprecision, however, the true states are unknown. Thus, computing a safe point is challenging, and the existing approaches prove inadequate. Note that, in general, true consensus is not possible when there is fixed imprecision in state values, however resilient consensus algorithms can fail to achieve \emph{approximate consensus} as defined in Section \ref{sec:intro}. For example, a centerpoint computed by a normal agent based on its neighbors' \emph{observed states} (due to imprecision) does not always lie in the convex hull of its normal neighbors' true states and, therefore, is not a safe point. Figure~\ref{fig:cp_imprecise} illustrates this situation.  

In Figure~\ref{fig:cp_imprecise}(a), we have six agents that are in the neighborhood of a normal agent $v$ (including $v$). One of the agents, $u$, is an adversary. Note that $v$ remains oblivious to the identity of the adversarial agent within its neighborhood. Each agent has a state in $\mathbb{R}^2$ accompanied by a region of imprecision, which we assume to be a square. The centerpoint region based on the observed states (represented by `$\bullet$') is highlighted in green, whereas the convex hull of normal agents' true states (indicated by `$\times$') is depicted as grey. Figure~\ref{fig:cp_imprecise}(b) reveals the challenge imprecision poses. The centerpoint region fails to remain entirely contained within the convex hull of the normal agents' true states. Consequently, agent $v$ may select a centerpoint (indicated by '$\circ$') that does not qualify as a safe point. As a result, $v$ may update its state by moving in a direction outside the convex hull of normal agents' true states, thus violating the safety condition of the resilient consensus. 

\begin{figure}[!h]
    \centering
\includegraphics[scale=.95]{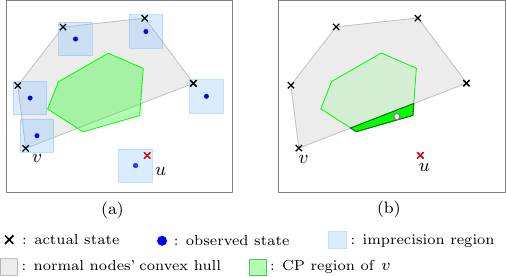}
    \caption{Centerpoint region based on the observed states (due to imprecision) is not contained entirely in the convex hull of normal agents' initial states.}
\label{fig:cp_imprecise}
\end{figure}

Figure~\ref{fig:res_con_imprecision} demonstrates examples of the run of the resilient consensus algorithm with and without the imprecision model. All six agents are pairwise adjacent, and there is one adversary. With no imprecision, all normal agents converge inside the convex hull of their initial states despite a single adversary, as shown by the state trajectories of agents in Figure~\ref{fig:res_con_imprecision}(a). However, with imprecision (square regions), the resilient consensus algorithm fails as normal agents do not remain inside their initial states' convex hull and continue moving farther away, as shown in Figure~\ref{fig:res_con_imprecision}(b). 

Thus, imprecision presents a significant obstacle to successfully operating resilient distributed algorithms. The primary hurdle arises from the \emph{difficulty in computing a safe point when dealing with imprecise observed states}. Thus, it is imperative to explore new approaches to ensure the accurate computation of safe points, even in the presence of imprecision. We address this problem in the next section.  

\begin{figure}[!h]
\centering
\begin{subfigure}[b]{0.24\textwidth}
\centering
\includegraphics[scale=0.1775]{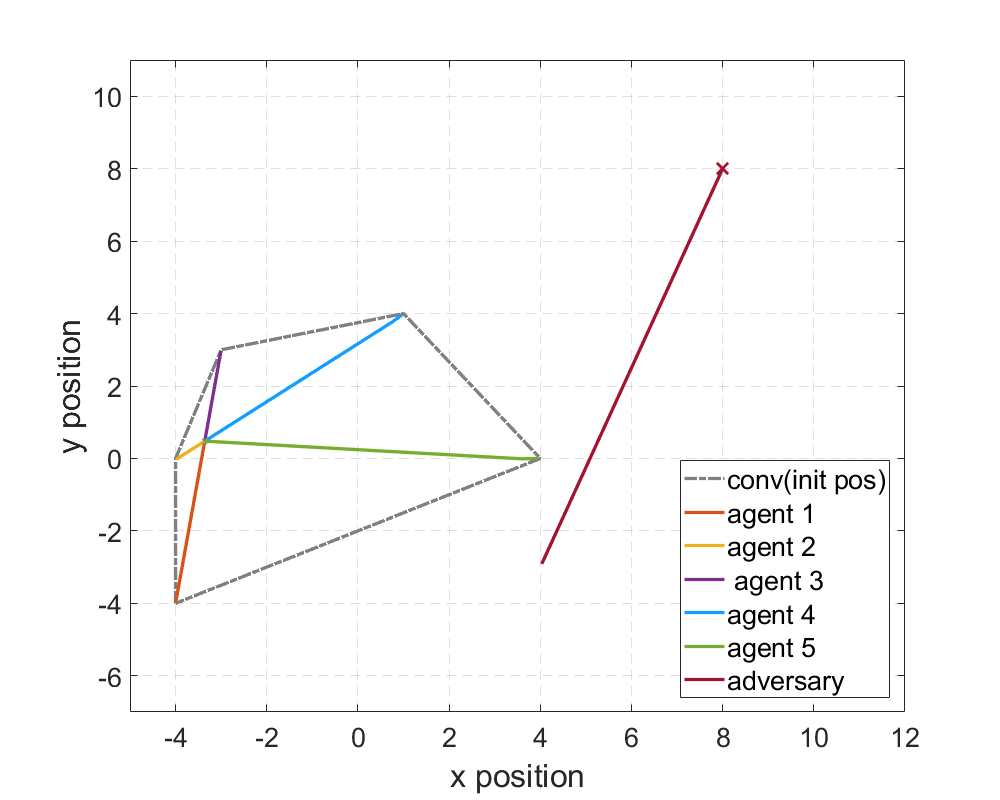}
\caption{}
\end{subfigure}
\begin{subfigure}[b]{0.24\textwidth}
\centering
\includegraphics[scale=0.1775]{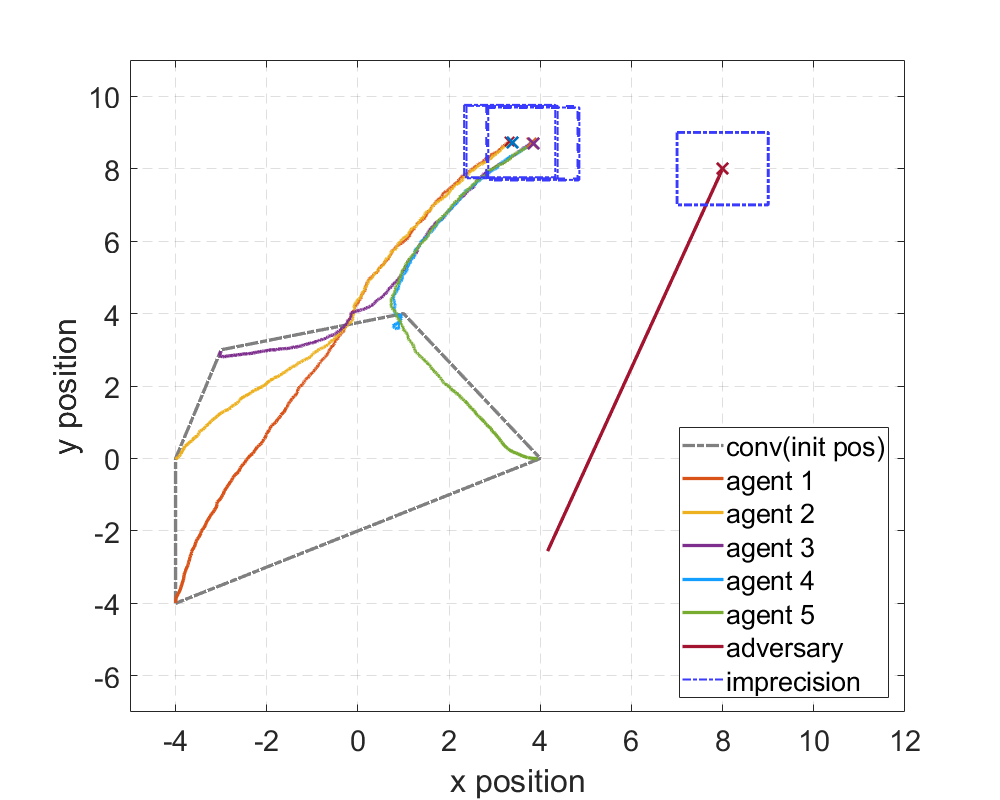}
\caption{}
\end{subfigure}
\caption{(a) Normal agents achieve resilient consensus with no imprecision. (b) Normal agents do not converge and move outside the convex hull of normal agents' initial states due to imprecision.}
\label{fig:res_con_imprecision}
\end{figure}

\section{Resilience in the Presence of Imprecision}
\label{sec:Invariant_hull}
In this section, we address the challenge of dealing with imprecision in resilient consensus algorithms and present our approach to mitigate its impact. Our main objective is to enable normal agents to \emph{compute safe points even when they encounter imprecise state values from their neighbors}. 

When a normal agent $v$ observes an imprecise state of its neighbor $u$, the agent $v$ essentially observes an imprecision region associated with $u$. This imprecision region contains the true state value of $u$, and is termed as a \emph{potential region} (as any value in this region can potentially be a true value of the agent). Consequently, agent $u$ effectively observes a collection of potential regions linked to its neighbors. By selecting one value (a point) from each potential region, we construct a \emph{potential configuration} for the given set of potential regions. It is worth noting that the true values of the agents represent just one potential configuration among an infinite set. Our focus lies in identifying the largest region contained within the convex hull of \emph{any} potential configuration from a given set of potential regions. We refer to this set as an \emph{invariant hull}. Importantly, an invariant hull, associated with a given set of potential regions is a subset of the convex hull of the true states of agents. At a high level, an invariant hull, pertaining to a specific set of potential regions, is akin to the convex hull of a given set of points. Figure~\ref{fig:invariant_hull} illustrates these concepts. Figure~\ref{fig:invariant_hull}(a) shows a set of potential regions (blue boxes) and the associated invariant hull. Figure~\ref{fig:invariant_hull}(b) shows one potential configuration (set of points in potential regions indicated by `$\times$') and the associated convex hull of the potential configuration (grey shaded region). Note that the invariant hull is contained in the convex hull of the potential configuration shown. Similarly, if we choose any other potential configuration, the invariant hull will be contained in the convex hull of such a configuration. 

\begin{figure}[!h]
\centering
\begin{subfigure}[b]{0.24\textwidth}
\centering
\includegraphics[scale=0.455]{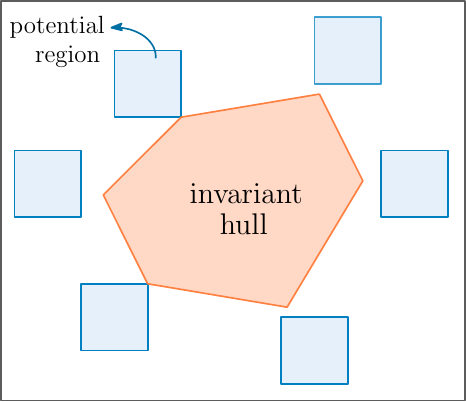}
\caption{}
\end{subfigure}
\begin{subfigure}[b]{0.24\textwidth}
\centering
\includegraphics[scale=0.455]{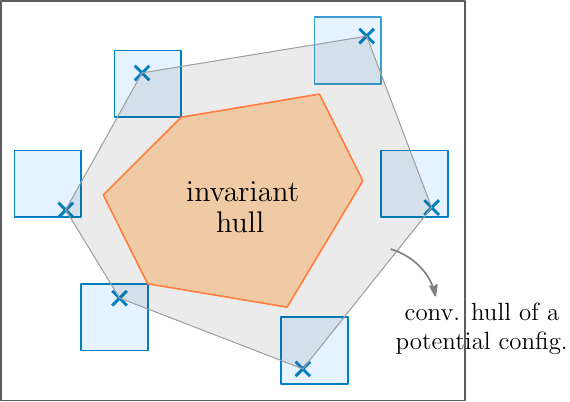}
\caption{}
\end{subfigure}
\caption{(a) Invariant hull of a set of potential regions. (b)~Invariant hull is a subset of the convex hull of an arbitrary potential configuration.}
\label{fig:invariant_hull}
\end{figure}

Next, we introduce some notations and formally define \emph{invariant hull}.

\begin{center}
\begin{tabular}{ l l}
 $x_v$ & Actual/true state of agent $v$; \\
 $r_v$ & Observed state of agent $v$; \\ 
 $B_v$ & potential region of agent $v$; \\  
 $B_V$ & Set of potential regions of agents in the set\\
 & $V =\{v_1,\cdots, v_n\}$.\\
\end{tabular}
\end{center}

\begin{definition}(Potential Configuration) For a given set of potential regions $B_V = \{B_1,\cdots,B_n\}$, a potential configuration is a set of points $P(B_V) = \{p_1,\cdots, p_n\}$ such that $p_v\in B_v$ for each $v$.
\end{definition}

Now, we define the notion of invariant hull.

\begin{definition} (Invariant Hull)
Consider a set of potential regions $B_V$. Let $\mathcal{P}(B_V)$ be the set of all possible potential configurations, and \texttt{Conv}(P) denotes the convex hull of a potential configuration $P\in\mathcal{P}(B_V)$. Then, the invariant hull of $B_V$ is defined as,

\begin{equation}
    \texttt{IHull}(B_V) = \bigcap\limits_{P\in\mathcal{P}(B_V)} \texttt{Conv}(P).
\end{equation}
\end{definition}

In simpler terms, for the invariant hull, we find all possible potential configurations of $B_V$, compute the convex hull of each such configuration, and finally compute the intersection of all such convex hulls. Consequently, the invariant hull is essentially a subset of the convex hull of any potential configuration of $B_V$, as Figure~\ref{fig:invariant_hull} illustrates. Also, the invariant hull is a convex set. 

Next, we provide a geometric characterization of the invariant hull. Then, in Section~\ref{sec:invariant_hull_algo}, we present an efficient method to compute the invariant hull of $B_V$.

\subsection{Characterization of Invariant Hull}
\label{sec:xzation_IHull}
First, we need to introduce notations to denote \emph{combinatorial families} of sets. 

\begin{itemize}
    \item Consider a set of potential regions \emph{$B_V$, then $B_V^{d+1}$ denotes a family of all $(d+1)$-member subsets of $B_V$} (i.e., subsets of $B_V$ consisting of $d+1$ potential regions). For example, consider $V = \{v_1,v_2,v_3,v_4\}$, the corresponding $B_V = \{B_1, B_2,B_3,B_4\}$, and $d =2$, then
\begin{align*}
    B_V^{3} = & \bm{\{}\;\{B_1,B_2,B_3\},\; \{B_1,B_2,B_4\}\; \{B_1,B_3,B_4\},\\
    & \{B_2,B_3,B_4\} \;\bm{\}}.
\end{align*}

\item Similarly, for $Q\in B_V^{d+1}$, the notation $Q^{d}$ denotes the family of all $d$-member subsets of $Q$. For example, in the above example, if $Q = \{B_1,B_2,B_3\}\in B_V^3$, then 
$$
Q^2 = \bm{\{} \{B_1, B_2\}, \; \{B_1, B_3\}, \; \{B_2, B_3\} \; \bm{\}}.
$$
\end{itemize}
So, in general, a \emph{superscript} on a set notation denotes the combinatorial family of the set. 

The following lemma identifies the essential property for a point to be in the invariant hull of a set of potential regions.

\begin{lem}
\label{lem:first}
Let \(Q \in B_V^{d+1} = \{q_1,..,q_{d+1}\}\) be a \((d+1)\)-member subset of potential regions.  If \(\texttt{IHull}(Q)\) is non-empty, then a point \(p \in \texttt{IHull}(Q)\) if and only if \(p\) has the following property: 
\begin{Properties}
\item For every hyperplane \(h\) passing through \(p\), at least one potential region \(q_i \subset Q\) is contained in each of the associated halfspaces of \(h\).
\end{Properties}
\end{lem}

\begin{proof} We prove the contrapositive. 
 For a point \(p\), if \(p \notin \texttt{IHull}(Q)\), then by definition there exists a particular potential configuration \(\sigma \in \mathcal{P}(Q)\) such that \(p\) and \(\texttt{Conv}(\sigma)\) are disjoint. Then let \(h\) denote a separating hyperplane, containing \(p\) for which all points of \(\sigma\) are contained in \(h^{inner}\), where \(h^{inner}\) is used to denote any halfspace of \(h\), with the opposite denoted \(h^{outer}\).  Since each of the \(d+1\) points of \(\sigma\) are selected from the \(d+1\) potential regions of \(Q\), it follows that every potential region of \(Q\) extends into \(h^{inner}\), and therefore no potential region is a subset of \(h^{outer}\). Thus, if \(p \notin \texttt{IHull}(Q)\), then \(p\) does not have Property 1.

Similarly, if a point \(p\) does not possess Property 1, then there is a hyperplane \(h'\) passing through \(p\) such that no potential region of \(Q\) is a subset of \(h^{outer}\). Then it is possible to select \(\sigma \in \mathcal{P}(Q)\) such that \(\sigma \subset h^{inner}\).  Since \(p \cap h^{inner} = \emptyset\), \(p\) and \(\sigma \) are disjoint, and thus \( p \notin \texttt{IHull}(Q)\). 
\end{proof}
\begin{cor}
\label{cor:4.2}
For \((d+1)\)-region set \(Q\), it follows from Lemma~\ref{lem:first} that  \(\texttt{IHull}(Q) = \texttt{Conv}(Q)\setminus\texttt{Conv}(\bigcup \limits_{q_i \in Q^d}q_i)\).  Here we have used the set notation of ``$A\setminus B$" to refer to the set of all elements of $A$ that are not elements of $B$. This can be seen by observing that no point with Property 1 may exist in the convex hull of a \(d\)-member subset of \(Q\). Any hyperplane through \(p\) and each of the \(d\) members may only contain the remaining region of \(Q\) in one or the other of its halfspaces, leaving one halfspace with no potential region as a subset.
\end{cor}

We now present an auxiliary lemma that will be used in proving Theorem~\ref{thm:IHull} characterizing the invariant hulls. 
\begin{lem}
\label{lem:second}
If \(\texttt{IHull}(B_V)\) is a convex polytope, and if point \(p\) is a vertex of \(\texttt{IHull}(B_V)\), then for all $d$-region subsets $f \in B_V^d$, $p \notin int(\texttt{Conv}(f))$, where $int(\texttt{Conv}(f))$ refers to the interior of $\texttt{Conv}(f)$, or $\texttt{Conv}(f)$ without its boundary points. 
\end{lem}

\begin{proof}
Let $f_{int}$ denote $int(\texttt{Conv}(f))$ in the following argument. Suppose \(f\) is a set of \(d\) potential regions of \(B_V\) such that \(p \in f_{int}\) and \(p\) is a vertex of \( \texttt{IHull}(B_V)\). There are two possibilities: 1) \( f_{int} \subset \texttt{IHull}(B_V)\), or 2) \( f_{int} \not\subset \texttt{IHull}(B_V)\).  In the first case, \(p \in int(\texttt{IHull}(B_V))\) and therefore $p$ is not a vertex.  In the second case, since $f_{int}$ is open, a $d$-point configuration of $f_{int}$ can always be found that includes $p_2 \ne p$ with $p_2 \in f_{int}\setminus \texttt{IHull}(B_V)$ such that $\texttt{Conv}(\texttt{IHull}(B_V) \cup p) \ne \texttt{Conv}(\texttt{IHull}(B_V) \cup p_2)$. But since $\texttt{IHull}(B_V) =\texttt{Conv}(\texttt{IHull}(B_V))$ is invariant for any potential configuration by its definition, this is a contradiction, and therefore the assertion of Lemma~\ref{lem:second} is true.
\end{proof}

We now show that the the invariant convex hull of \(n\) potential regions is the convex hull of the union of invariant hulls of each of the  \({n}\choose{d+1}\) subsets of \(B_V^{d+1}\). First, we state a generalization of a theorem of Caratheodory to sets in \(\mathbb{R}^d\), that will be used in our proof. 
\begin{theorem}\cite{BARANY1982141}
\label{thm:Cara}
For a family of sets $K\in \mathbb{R^d}$, with $|K|\ge d+1$,  $\bigcup\limits_{k_i \in K^{d+1}}\texttt{Conv}(k_i) = \texttt{Conv}(K)$.
\end{theorem}
\textit{Note:  Here and throughout the paper, if $K$ is a family of point sets, then $\texttt{Conv}(K) = \texttt{Conv}(\bigcup \limits_{k \in K} K)$}.
\\
In other words the convex hull of every point contained in the family of sets is equivalent to the union of the convex hulls of all of its $(d+1)$-member subsets.
Now, we state one of the main results of this section.
\begin{theorem}
\label{thm:IHull}
Let \(B_V= \{B_1,\cdots,B_n\}\) be a collection of \(n\) potential regions of \(\mathbb{R}^d\) and \(B_V^{d+1}\) denote the family of \((d+1)\)-member subsets of \(B_V\). Then, 
\begin{equation}
\texttt{IHull}(B_V) = \texttt{Conv}(\bigcup \limits_{Q\in B_V^{d+1}} \texttt{IHull}(Q)) .  
\end{equation}
\end{theorem}
\begin{proof}
We will show that all vertices of \(\texttt{IHull}(B_V)\) are contained in \(\bigcup \limits_{Q\in B_V^{d+1}} \texttt{IHull}(Q)\).  If \(p\) is a vertex of \( \texttt{IHull}(B_V)\), clearly \(p \in \texttt{Conv}(B_V)\). By Lemma~\ref{lem:second}, \(p\) cannot be contained in the interior of the convex hull of any \(d\)-member subset of \(B_V)\). By Theorem~\ref{thm:Cara},  the convex hull of \(B_V\) in \(\mathbb{R}^d\) is equal to the union of the convex hulls of its \((d+1)\)-member subsets.  From Lemma~\ref{lem:first} and Lemma~\ref{lem:second}, it follows that \(\bigcup \limits_{Q\in B_V^{d+1}} \texttt{IHull}(Q)\) precisely contains all points that are in the interior of \(B_V\) but not interior to the convex hull of any \(d\)-member subset.  This implies that \(p \in \bigcup \limits_{Q\in B_V^{d+1}} \texttt{IHull}(Q)\).  Furthermore for any \(Q\), every point \(x \in \texttt{IHull}(Q)\) is trivially a point in \(\texttt{IHull}(B_V)\), since \(Q \subset B_V\) and any point configuration of \(B_V\) has a subset of points from the potential regions of \(Q\). Since $\bigcup \limits_{Q\in B_V^{d+1}} \texttt{IHull}(Q)$ contains all vertices of $\texttt{IHull}(B_V)$ and contributes no points exterior to $\texttt{IHull}(B_V)$, it follows that \(\texttt{IHull}(B_V) = \texttt{Conv}(\bigcup \limits_{Q\in B_V^{d+1}} \texttt{IHull}(Q)) \).
\end{proof}

\subsection{Computing the Invariant Hull}
\label{sec:invariant_hull_algo}

We now outline a procedure for computing the invariant hull of \(B_V\), given that \(|B_V| > d+1\). The procedure involves computing the equation of a tangent hyperplane equation to \(d\) potential regions of \(\mathbb{R}^d\). In the case that potential regions are polygons in \(\mathbb{R}^2\), linear-time algorithms have been developed to find \textit{outer} tangent lines, that is the tangent lines that have all participating polygons on one side ~\cite{CommonTangents2019}. 
\\
Let \(Q = \{q_1,q_2,\cdots,q_{d+1}\} \in B_V^{d+1}\). As stated in Corollary~\ref{cor:4.2}, \(\texttt{IHull}(Q)\)is equivalent to \(\texttt{Conv}(Q)\setminus \texttt{Conv}(Q^d)\). A simple approach is to compute the vertices of \(\texttt{IHull}(Q)\):
\begin{enumerate}
\item Select \(q_i \in Q\). 
\item For each \(j = 1,2,\cdots,d+1\), and each \(f_j \in Q^d\) containing \(q_i\), compute hyperplane \(h_j\) tangent to the \(d\) members of \(f_j\) such that \(f_j \subset h^{outer}\) and \(Q\setminus f \subset h^{inner}\). 
\item Compute point of intersection:\(\{x \in \mathbb{R}^d: h_1 \cdot x = h_2 \cdot x = \cdots =h_{d+1} \cdot x\}\). Label \(x\) as \(w_i\).
\item Repeat steps 1-3 for all remaining \(q_i \in Q\). 
\item \(\texttt{IHull}(Q) = \texttt{Conv}(\bigcup \limits_{\forall q_i \in Q} w_i)\).
\end{enumerate}
After repeating steps 1-5 for all \(Q \in B_V^{d+1}\), the invariant hull of \(B_V\) is computed with any convex hull routine  as:  \[\texttt{IHull}(B_V) = \texttt{Conv}(\bigcup \limits_{Q\in B_V^{d+1}} \texttt{IHull}(Q)), \]
according to Theorem~\ref{sec:invariant_hull_algo}.

\begin{figure*}[!h]
\centering
\begin{subfigure}[b]{0.133\textwidth}
\centering
\includegraphics[scale=0.29]{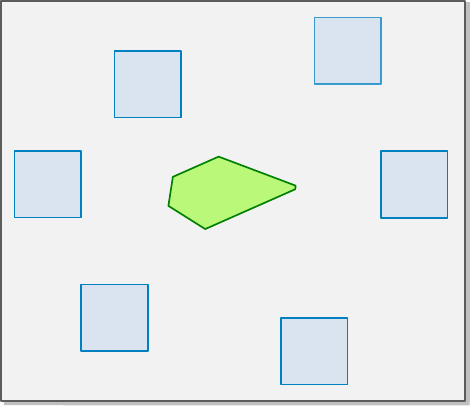}
\caption{}
\end{subfigure}
\begin{subfigure}[b]{0.133\textwidth}
\centering
\includegraphics[scale=0.29]{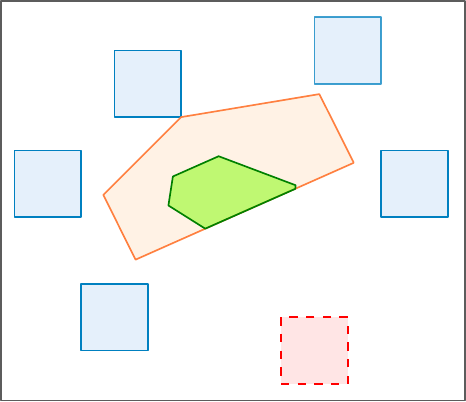}
\caption{}
\end{subfigure}
\begin{subfigure}[b]{0.133\textwidth}
\centering
\includegraphics[scale=0.29]{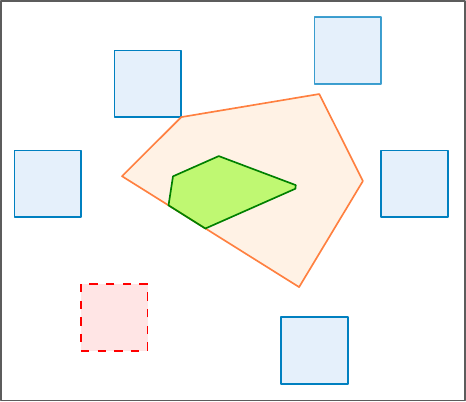}
\caption{}
\end{subfigure}
\begin{subfigure}[b]{0.133\textwidth}
\centering
\includegraphics[scale=0.29]{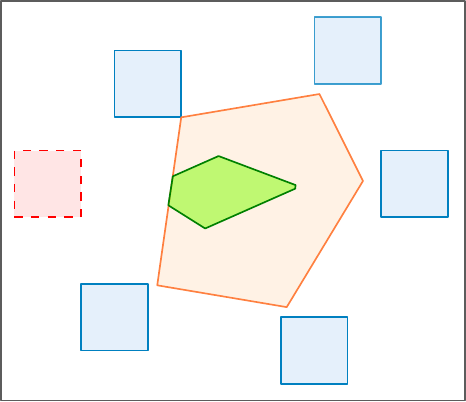}
\caption{}
\end{subfigure}
\begin{subfigure}[b]{0.133\textwidth}
\centering
\includegraphics[scale=0.29]{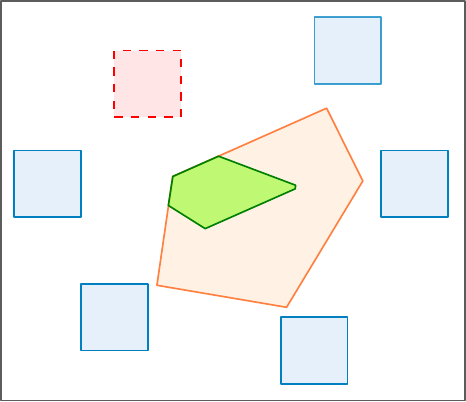}
\caption{}
\end{subfigure}
\begin{subfigure}[b]{0.133\textwidth}
\centering
\includegraphics[scale=0.29]{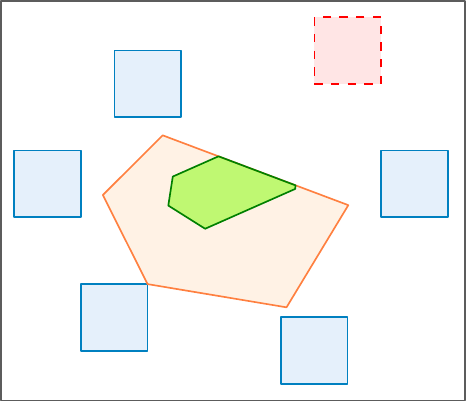}
\caption{}
\end{subfigure}
\begin{subfigure}[b]{0.133\textwidth}
\centering
\includegraphics[scale=0.29]{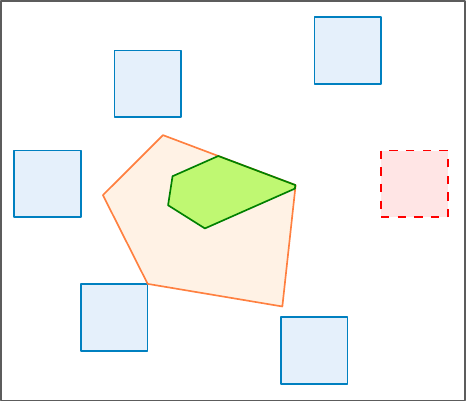}
\caption{}
\end{subfigure}
\caption{(a) Centerpoint region based on the intersection of the invariant hulls of subsets of potential regions. In each of (b)--(g), the green shaded region is contained entirely in the invariant hull of the five potential regions.}
\label{fig:CP}
\end{figure*}

 \section{Resilient Consensus through Centerpoint using Invariant Hulls of Potential Regions}
 \label{sec:CPIH}

 In this section, we present a way for a normal agent $v$ to compute a safe point despite access to only imprecise states of neighbors. Moreover, the proposed approach allows $F$ number of adversaries in the neighborhood of $v$, where $F \le \frac{N_v} {d+1} - 1$.  The core of our method relies on the intersection of the invariant hulls of potential regions.

 First, consider the case with no imprecision, i.e., normal agents observe true states of their neighbors (as discussed in \cite{abbas2022resilient,li2022byzantine}). In this ideal situation, a safe point for a normal agent $v$ is effectively the centerpoint of its neighbors' states, provided the number of adversarial agents in the neighborhood of $v$ is $F \le \frac{N_v} {d+1} - 1$. This proposition is rooted in geometric principles. Specifically, a centerpoint of a given set of $N_v$ points in $\mathbb{R}^d$ lies within the convex hull of any subset comprising more than $\frac{d}{d+1}N_v$ of those points \cite{pach2011combinatorial,matousek2013lectures,ray2007optimal}. To put it practically, if we select an arbitrary subset of neighbors of $v$---comprising more than $\frac{d}{d+1}N_v$ neighbors---and compute the convex hull of their states, a centerpoint is guaranteed to lie within this hull. Since one such choice of more than $\frac{d}{d+1}N_v$ neighbors of a normal agent $v$ consists of only normal neighbors of $v$, a centerpoint of the states of agent $v$'s neighbors also lies within the convex hull of \emph{only} the normal neighbors' states. 

 In cases where the observed states are imprecise, the approach to determining a safe point for a normal agent $v$ must be modified. As elaborated in Section~\ref{sec:Res_do_not_work}, the inherent imprecision precludes the use of convex hulls calculated from observed states. Instead of knowing a true state of the neighbor, $v$ only knows the potential region of the neighbor (containing the true state).  So, instead of computing convex hulls of subsets of neighbors' observed states, a normal agent computes the invariant hulls of subsets of potential regions of neighbors. Assume that a normal agent $v$ has $N_v$ neighbors, of which at most $\frac{N_v}{d+1}-1$ are adversarial. For any subset of neighbors containing more than $\frac{d}{d+1}{N_v}$ neighbors, agent $v$ computes the invariant hull of the potential regions of the corresponding neighbors, and then computes the intersection of all such invariant hulls. This resulting intersection serves as a modified `centerpoint region', constructed from the invariant hulls of potential regions. Importantly, every point in this new centerpoint region is guaranteed to lie within the convex hull of the true states of agent $v$'s normal neighbors, and is thus a safe point (despite imprecise observed states). We illustrate this through an example below. 

 \emph{Example:} Consider a set of six potential regions in a plane, each corresponding to an agent. Note that the true states of these agents are `hidden' within these potential regions. Among these six, one potential region belongs to an adversarial agent, though we don't know which one. The green-shaded region in Figure~\ref{fig:CP}(a) is the centerpoint region constructed using the intersection of invariant hulls of subsets of potential regions. To see this, consider a subset of five potential regions---calculated as $\frac{d}{d+1}N + 1 = 5$---and determine their invariant hull.  As demonstrated in Figures~\ref{fig:CP}(b)--\ref{fig:CP}(g), the green centerpoint region is consistently enclosed within the invariant hull, regardless of which five potential regions are chosen. Notably, one of these choices in Figures~\ref{fig:CP}(b)--\ref{fig:CP}(g) includes the actual adversary. In this case, the computed invariant hull---by definition---remains a subset of the convex hull formed by the true states of normal agents. Therefore, the points in the green-shaded centerpoint region, which is contained entirely in the invariant hull, are assured to be safe points.

Next, we outline the steps of the state update for each normal agent. An essential step is the computation of a safe point using the idea of invariant hulls, as described previously. The steps of the algorithm, which we call \emph{centerpoint of invariant hulls (CPIH)} method, are as follows:

 For each node \(v \in V_n\). 
\begin{enumerate}
\item At time \(t\) compute \(k = \frac{dN_v}{d+1}+1 \).
\item For \(C \in B_V^k\), compute \(\texttt{IHull}(C)\). ( Compute invariant hull of each \(k\)-tuple of \(B_V\)).
\item If \( \bigcap \limits_{C \in B_V^k} \texttt{IHull}(C) = \emptyset\), set \(x_v(t+1) = x_v(t)\).  Otherwise, proceed to step 4. (Do nothing at time \(t\) if safe region does not exist).
\item Select safe point \(p_v(t) \in \bigcap \limits_{C \in B_V^k} \texttt{IHull}(C)\).   
\item \(x_v(t+1) = \alpha_v(t)p_v(t)+(1-\alpha_v(t))x_v(t)\).  
\end{enumerate}
The value of \(\alpha_v(t)\) is a dynamic weight that may be chosen in range \([0,1]\) according to the application.

The CPIH algorithm identifies a region that is a generalization of the centerpoint region in that the property of a CPIH safe point has identical properties to that of of a centerpoint with compact sets replacing the role of points.  Thus each halfspace of a hyperplane that intersects a CPIH safe point contains at least \(\lfloor \frac{N}{d+1} \rfloor\) compact sets (potential regions of \(B_V\)). As such, every \( \frac{dN}{d+1}+1\) potential regions contain a CPIH safe point. 

\subsection{Simulations and Discussion}
\begin{figure*}[!h]
\centering
\begin{subfigure}[b]{0.28\textwidth}
\centering
\includegraphics[scale=0.2]{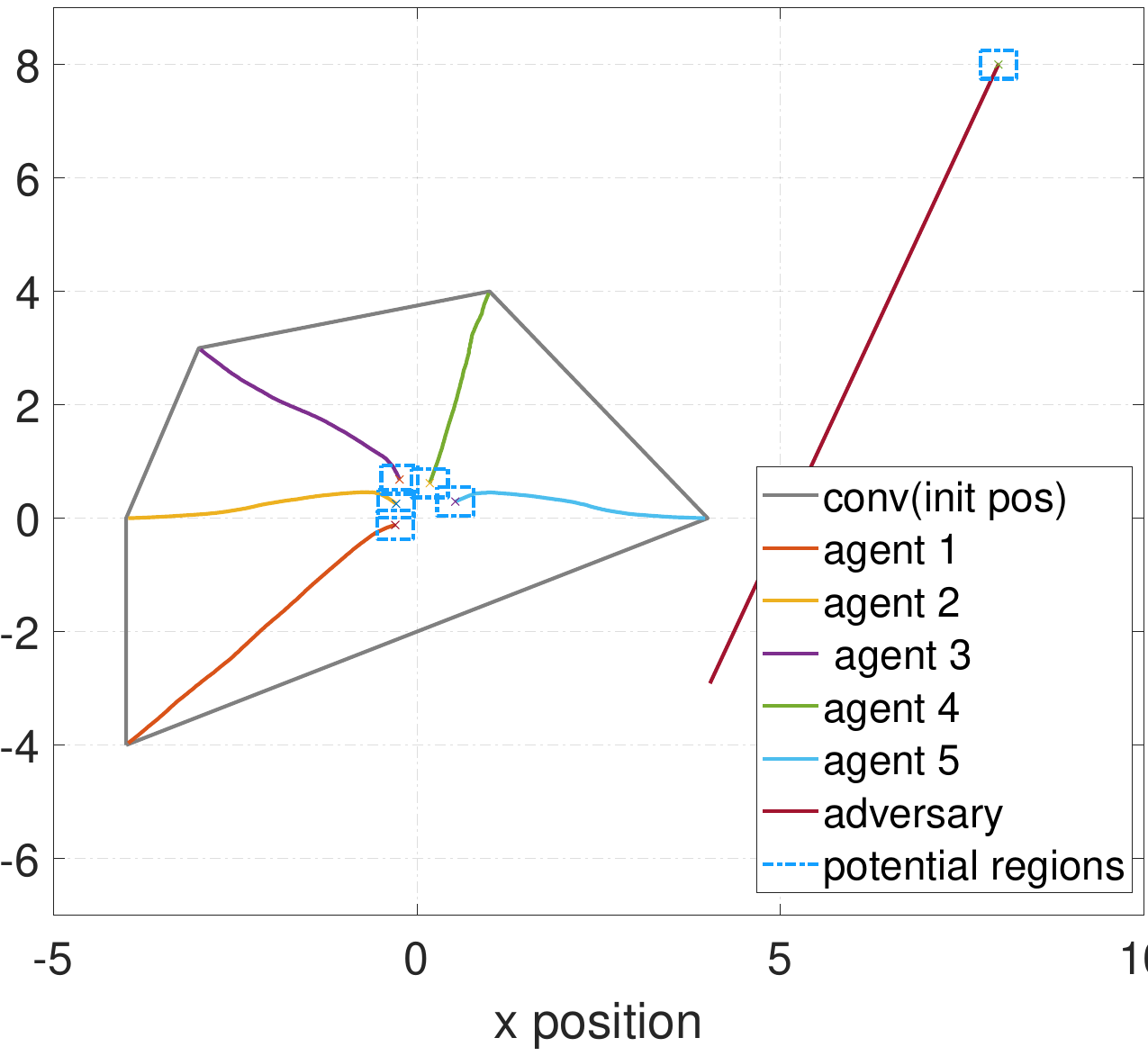}
\caption{$\delta = 0.5$}
\end{subfigure}
\begin{subfigure}[b]{0.28\textwidth}
\centering
\includegraphics[scale=0.2]{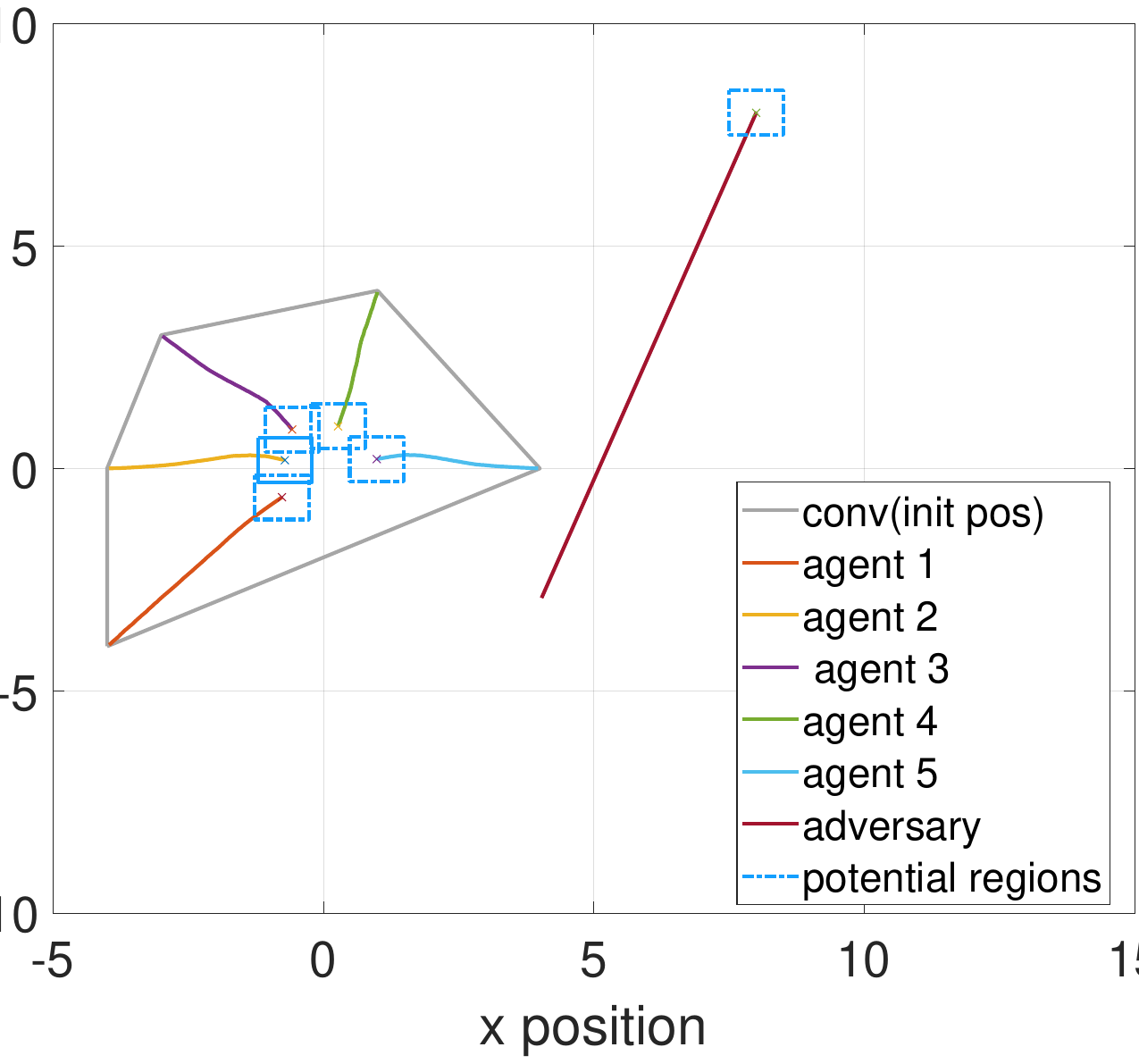}
\caption{$\delta = 1$}
\end{subfigure}
\begin{subfigure}[b]{0.28\textwidth}
\centering
\includegraphics[scale=0.2]{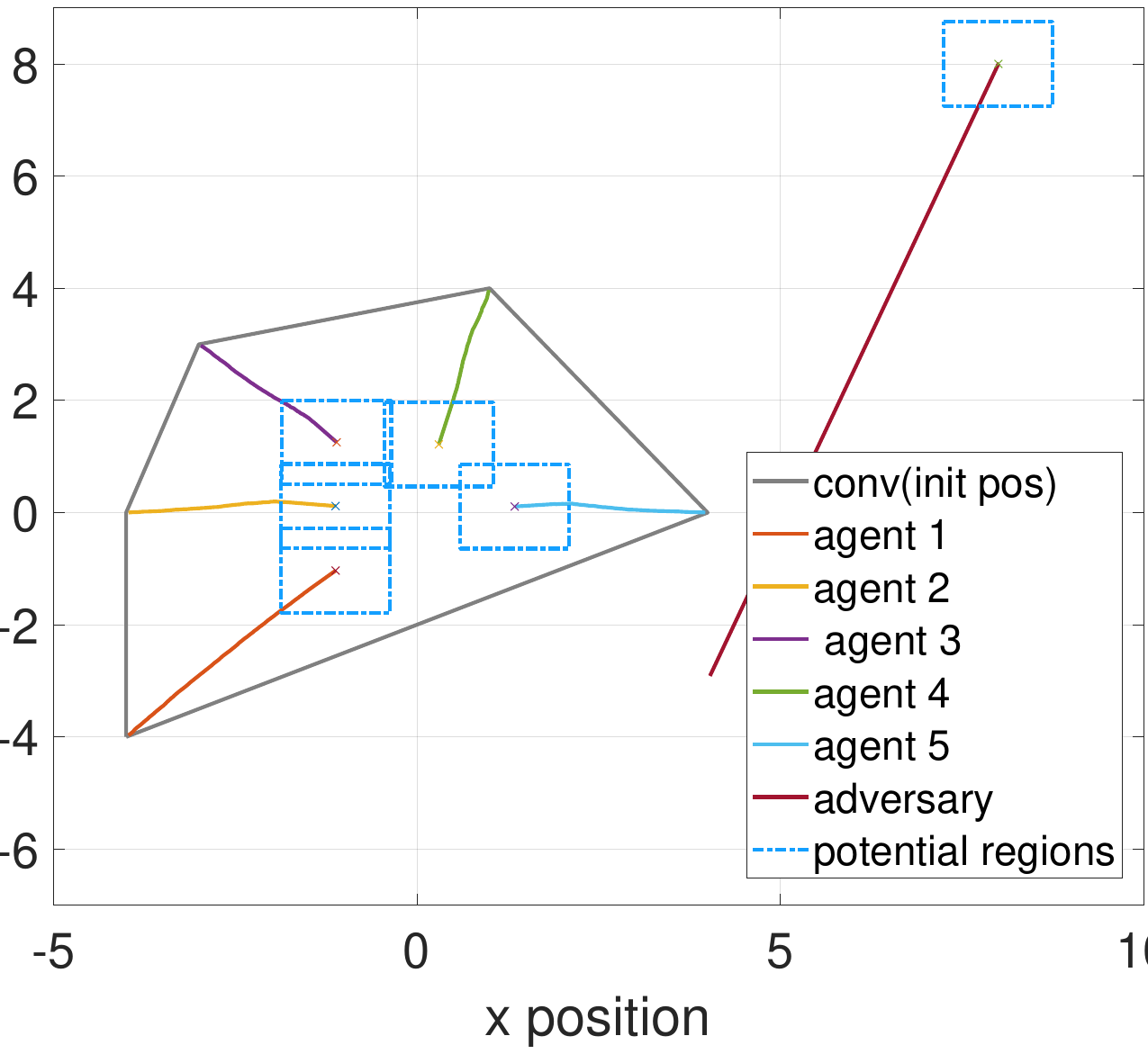}
\caption{$\delta = 1.5$}
\end{subfigure}
\caption{Implementation of CPIH algorithm. Normal agents remain inside the convex hull of their initial states and converge to a degree that depends on the level of state imprecision, as quantified by the parameter $\delta$.}
\label{fig:good_plots}
\end{figure*}
In order to verify that the CPIH Algorithm succeeds where the ordinary centerpoint-based consensus algorithm fails (as illustrated in Figure~\ref{fig:res_con_imprecision}), we executed a series of simulations using the above algorithm. For simplicity, potential regions were chosen to be a square of width \(\delta\) centered around each node, and $G$ was chosen to be a complete graph.  At each time step, agents received a uniformly random state estimate within the potential regions of each of its neighbors. For illustrative purposes, we chose a small range of values for \(\delta\). For each value of \(\delta\), a  simulation was run for $5000$ time steps. Figure~\ref{fig:good_plots} shows the results.

As may be inferred from Step 3 of the CPIH algorithm, this procedure will not, in general, result in uniform convergence of agents states. Instead, nodes will ``cluster'' into a region whose volume depends on the magnitude of \(\delta\).  While this may appear to be an unsatisfactory result, it nonetheless assures that nodes will remain within their initial convex hull, whereas this is not the case for the resilient consensus algorithms that do not account for imprecision. As a result, CPIH algorithm presents a tradeoff between network safety and convergence precision, which is of significance on its own, especially when safety is of high priority in a setting with imprecision. In real world situations, imprecision typically decreases with proximity. Future developments of CPIH could achieve a convergent consensus  under such an imprecision model.




\section{Conclusion}
\label{sec:con}
In this paper, we have demonstrated that traditional resilient consensus algorithms can fail when they do not account for state imprecision. To address this, we proposed an uncertainty-aware geometric solution ensuring resilience against adversaries and state imprecision. Specifically, we extended the centerpoint-based resilient consensus algorithm by introducing the concept of an 'invariant hull,' which we call the CPIH algorithm. Through illustrative simulations, we have showcased the effectiveness of our algorithm under varying levels of uncertainty. In the future, we aim to explore the resilience-accuracy trade-off further by determining the conditions under which the invariant hull of imprecise regions becomes empty. 

\bibliographystyle{IEEEtran} 
\bibliography{references}  

\end{document}